\documentclass[12pt]{article}
\usepackage{amsmath}
\DeclareMathAlphabet{\mathcal}{OMS}{cmsy}{m}{n}
\usepackage{amsthm}
\usepackage{amssymb}
\usepackage{txfonts}
\usepackage{mathrsfs}
\usepackage{arydshln}
\usepackage{blkarray}
\usepackage{graphicx}
\usepackage{cite}
\usepackage{bm}
\usepackage{color}
\usepackage{xcolor}
\usepackage{polynom}
\usepackage{framed}
\usepackage{tikz}
\usetikzlibrary{tikzmark,decorations.pathreplacing,calligraphy}
\tikzset{every picture/.style={remember picture}}
\usepackage{multirow}
\usepackage{hyperref}
\hypersetup{colorlinks=false,pdfborder={0 0 1},
  linkbordercolor={1 0 0},citebordercolor={0 1 0},
  urlbordercolor={0 1 1},hypertexnames=false}
\DeclareMathOperator{\rank}{rank}

\usepackage{multirow}  
\usepackage{booktabs}
\usepackage{bm}
\usepackage{enumitem}
\setlist[itemize]{itemsep=2pt, parsep=0pt, topsep=4pt, partopsep=0pt}

\AtBeginDocument{
	\setlength{\abovedisplayskip}{6pt}
	\setlength{\belowdisplayskip}{6pt}
	\setlength{\abovedisplayshortskip}{4pt}
	\setlength{\belowdisplayshortskip}{4pt}
}

\definecolor{red}{rgb}{1,0,0}
\definecolor{blue}{rgb}{0,0,1}
\definecolor{darkblue}{rgb}{0,0,0.4}

\begin{document}
\newtheorem{Def}{Definition}[section]
\newtheorem{example}{Example}[section]
\newtheorem{proposition}[Def]{Proposition}
\newtheorem{theorem}[Def]{Theorem}
\newtheorem{lemma}[Def]{Lemma}
\newtheorem{remark}{\noindent\mbox{Remark}}[section]
\newtheorem{corollary}[Def]{Corollary}
\newtheorem{problem}{Problem}

\newcommand{\F}{\mathbb{F}}
\newcommand{\cK}{\mathcal{K}}
\newcommand{\cS}{\mathcal{S}}
\newcommand{\Ker}{\rm Ker}
\newcommand{\Soc}{\rm Soc}
\newcommand{\lra}{\longrightarrow}
\newcommand{\ra}{\rightarrow}
\newcommand{\add}{{\rm add\, }}
\newcommand{\gd}{{\rm gl.dim\, }}
\newcommand{\End}{{\rm End\, }}
\newcommand{\overpr}{$\hfill\square$}
\newcommand{\rad}{{\rm rad\,}}
\newcommand{\soc}{{\rm soc\,}}
\renewcommand{\top}{{\rm top\,}}
\newcommand{\fdim}{{\rm fin.dim}\,}
\newcommand{\fidim}{{\rm fin.inj.dim}\,}
\newcommand{\gldim}{{\rm gl.dim}\,}
\newcommand{\cpx}[1]{#1^{\bullet}}
\newcommand{\D}[1]{{\mathscr D}(#1)}
\newcommand{\Dz}[1]{{\rm D}^+(#1)}
\newcommand{\Df}[1]{{\rm D}^-(#1)}
\newcommand{\Db}[1]{{\mathscr D}^b(#1)}
\newcommand{\C}[1]{{\mathscr C}(#1)}
\newcommand{\Cz}[1]{{\rm C}^+(#1)}
\newcommand{\Cf}[1]{{\rm C}^-(#1)}
\newcommand{\Cb}[1]{{\mathscr C}^b(#1)}
\newcommand{\K}[1]{{\mathscr K}(#1)}
\newcommand{\Kz}[1]{{\rm K}^+(#1)}
\newcommand{\Kf}[1]{{\rm K}^-(#1)}
\newcommand{\Kb}[1]{{\mathscr K}^b(#1)}
\newcommand{\modcat}[1]{#1\mbox{{\rm -mod}}}
\newcommand{\stmodcat}[1]{#1\mbox{{\rm -{\underline{mod}}}}}
\newcommand{\pmodcat}[1]{#1\mbox{{\rm -proj}}}
\newcommand{\imodcat}[1]{#1\mbox{{\rm -inj}}}
\newcommand{\opp}{^{\rm op}}
\newcommand{\otimesL}{\otimes^{\rm\bf L}}
\newcommand{\rHom}{{\rm\bf R}{\rm Hom}\,}
\newcommand{\pd}{{\rm pd}}
\newcommand{\Hom}{{\rm Hom \, }}
\newcommand{\Coker}{{\rm coker}\,\,}
\newcommand{\Ext}{{\rm Ext}}
\newcommand{\im}{{\rm im}}
\newcommand{\Cone}{{\rm Cone}}
\newcommand{\fini}{{finitistic}}
\newcommand{\proj}{{\rm  proj}}
\newcommand{\al}{{\alpha}}
\newcommand{\Tr}{{\rm  Tr}}
\newcommand{\tr}{{\rm  tr}}
\newcommand{\Norm}{{\rm  Norm}}
\newcommand{\Tor}{{\rm  Tor}}
\newcommand{\ad}{{\rm  add}}
\newcommand{\StHom}{{\rm \underline{Hom} \, }}
\renewcommand{\thefootnote}{\arabic{footnote}}

%\numberwithin{equation}{section}

%\tableofcontents
%\newpage

\renewcommand{\theequation}{\arabic{section}.\arabic{equation}}
\pagenumbering{arabic}{\Large \bf

\begin{center}
Puncturing-inclusion constructions of optimal Ferrers diagram rank-metric codes using semifield codes
\end{center}}

\centerline{ {\sc Dengming Xu$^1$\qquad Chengming Duan$^2$ \qquad Yong Jin$^3$ \qquad Yuhao Ning$^4$   }}
\begin{abstract}
Ferrers diagram rank-metric codes are essential components of the multilevel construction of constant dimension codes for random network coding. This paper aims to construct optimal Ferrers diagram rank-metric codes, focusing on the diagram \([3,3,3,5]\) with rank distance \(3\). We construct an optimal \([[3,3,3,5],6,3]_q\) Ferrers diagram rank-metric code for every prime power \(q\), thereby answering an open problem posed by Etzion et al. in 2016~\cite{Etzion2016}. The construction first builds a \([3\times4,4,3]_q\) MRD code and then proves that the Delsarte dual of its puncturing  on  the last  column contains a \([3\times3,3,3]_q\) MRD code, which  comes from semifield codes. We further construct optimal Ferrers diagram rank-metric codes for two infinite families of Ferrers diagrams: the first relies on the \([[3,3,3,5],6,3]_q\) construction and holds for every prime power \(q\), while the second is binary and obtained directly by the puncturing-inclusion construction.
\end{abstract}

\vskip-1cm

\renewcommand{\thefootnote}{\alph{footnote}}
\setcounter{footnote}{-1} \footnote{$^1$ Corresponding author, Sino-European Institute of Aviation Engineering,
Civil Aviation University of China, 300300 Tianjin,
PR China. E-mail:
xudeng17@163.com.
}
\setcounter{footnote}{-1} \footnote{$^2 $ College of Science,
Civil Aviation University of China, 300300 Tianjin,
PR  China. E-mail: 19160306346@163.com
}
\setcounter{footnote}{-1} \footnote{$^3 $ College of Science,
	Civil Aviation University of China, 300300 Tianjin,
	PR  China. E-mail: kingmeng@126.com
}
\setcounter{footnote}{-1} \footnote{$^4 $ College of Science,
	Civil Aviation University of China, 300300 Tianjin,
	PR  China. E-mail: ningyuh00@163.com
}

\section{Introduction}\label{sec:introduction}

Constant dimension codes are collections of subspaces of a fixed dimension, equipped with the subspace distance. Their role in error correction for random network coding was established by Koetter and Kschischang~\cite{KoetterKschischang2008}. A basic connection with matrix codes is as follows: lifting a matrix \(A\) to the row space of \([I\mid A]\) doubles the rank distance into the subspace distance. This connection makes maximum rank-distance (MRD) codes, originating in the work of Delsarte~\cite{delsarte1978bilinear}, natural ingredients in constructions of large constant dimension codes.

The multilevel construction of Etzion and Silberstein~\cite{EtzionSilberstein2009} can be viewed as a generalization of the lifted MRD code construction, in which Ferrers diagram rank-metric codes (FDRMCs) play an essential role; we refer the reader to~\cite{EtzionSilberstein2009} for the details.  For a Ferrers diagram \(\mathcal{F}\) and a prescribed minimum rank distance \(d\), the Ferrers--Singleton bound gives an upper bound \(\nu_{\min}(\mathcal{F},d)\) on the dimension of a linear code. A code attaining this bound is called optimal.  Etzion and Silberstein conjectured that the bound is attainable for every diagram and finite field~\cite{EtzionSilberstein2009}.
Since the conjecture was proposed, the construction of optimal Ferrers diagram rank-metric codes has received considerable attention. The main constructions in the literature are based on systematic MRD codes \cite{Etzion2016,LiuChangFeng2019,AntrobusGluesingLuerssen2019,NeriStanojkovski2024,ZhangGe2019},  on MDS codes \cite{Etzion2016}, and combinatorial methods \cite{Etzion2016,LiuChangFeng2019},
for more results, please see \cite{AntrobusGluesingLuerssen2019} and the references therein.

Recently, Beeloo--Sauerbier Couv\'ee and Neri reduced the attainability problem to irreducible diagrams and propose the puncturing--inclusion MRD conjecture\cite{BeelooNeri2026} . For integers \(1<d<n\), this conjecture asks for an
$
[(n-1)\times n,n(n-d),d]_q\text{ MRD code}
$
whose puncturing on one column is contained in an
$
[(n-1)\times(n-1),(n-1)(n-d+1),d-1]_q\text{ MRD code}.
$
%The core idea is to first construct a larger rectangular MRD code and then ensure that its puncturing on one column is contained in a smaller MRD code of distance \(d-1\). 
This puncturing--inclusion property builds a bridge between classical MRD codes and optimal Ferrers diagram rank-metric codes: by reducing to irreducible diagrams, it reduces the Etzion--Silberstein conjecture for a special class of Ferrers diagrams to a statement about rectangular MRD codes (see Theorem~7.5 in \cite{BeelooNeri2026}).

The puncturing--inclusion MRD conjecture is false in general. Recently, Prajapati~\cite{Prajapati2026} proved that, over \(\mathbb{F}_2\), the diagram \([4,4,4,4,6]\) with minimum rank distance \(3\) has Ferrers--Singleton bound \(12\), but maximum linear dimension \(11\). 
 However, determining which Ferrers diagrams admit optimal Ferrers diagram rank-metric codes remains an interesting problem worthy of further investigation, due to their crucial role in the multilevel construction of constant dimension codes.

The paper aims to study the existence of optimal Ferrers diagram rank-metric codes via the puncturing--inclusion method, we   focus on the diagram \([3,3,3,5]\) with rank distance \(3\).  To this end
, we first construct a \([3\times4,4,3]_q\) MRD code $\mathcal C$. We translate the construction of this MRD code into the problem of finding a determinant over \(\mathbb{F}_q\) involving three variables such that the determinant is nonzero whenever these variables are not all zero. This further reduces the problem to finding a ternary polynomial over \(\mathbb{F}_q\) with no nonzero root. It should be
 noted that this discussion needs to be carried out separately in the cases of characteristic \(2\) and characteristic not equal to \(2\). In odd characteristic, one needs a nonsquare element, whereas in characteristic \(2\) every element is a square, so the method does not apply directly and requires a separate treatment.

Furthermore, we  find a three-dimensional semifield code contained in the Delsarte dual of the puncturing of \(\mathcal{C}\) on its last column. This problem is then translated into that of finding two bases (over \(\mathbb{F}_q\)) in a semifield \((\mathbb{F}_{q^3},\circ)\) satisfying the corresponding orthogonality equations.  Then, using properties of semifield operations, we further reduce these orthogonality equations to finding a basis \((1,x,y)\) of \(\mathbb F_{q^3}\)  over $\mathbb F_q$ such that
$$
x\circ y-y\circ x \in \mathbb F_q\setminus\{0\}.
$$

In this way, we  prove that the puncturing--inclusion MRD conjecture holds for \((n,d)=(4,3)\), and hence proving the following theorem
(see Theorem \ref{thm:all}).

\noindent{\bf Theorem A}
{\it For every prime power \(q\), there exists an optimal
	\([[3,3,3,5],6,3]_q\) Ferrers diagram rank-metric code.}

%It should be noted that the above proof does not cover the three small-field cases \(q=2,3,4\); for these cases, we give the corresponding optimal codes explicitly ( see Remarks \ref{rem:3-2},\ref{rem:4-1}) 

The existence of optimal
\([[3,3,3,5],6,3]_q\) Ferrers diagram rank-metric code was proposed as an open problem in \cite{Etzion2016}.  It has remained unresolved for a long time. Different from the existing constructions based on systematic MRD codes and MDS codes, or the combinatorial methods, we combine the puncturing-inclusion construction with semifield codes. This offers a new perspective for constructing optimal Ferrers diagram rank-metric codes.

As an application of Theorem A,  we obtain optimal Ferrers diagram rank-metric codes for an infinite family of diagrams by combination constructions (see Theorem~\ref{thm:general}):

\noindent{\bf Theorem B}
{\it  Let $n\ge1$ and write $n=6t+r$, where $t\ge0$ and $0\le r\le5$. Define
\begin{align*}
	\mathcal F_n=\big[&3,3,3,5,
	\underbrace{11,\cdots,11}_{6},
	\underbrace{17,\cdots,17}_{6},\cdots,
	\underbrace{5+6t,\cdots,5+6t}_{6},\underbrace{5+6(t+1),\cdots,5+6(t+1)}_{r}\big],
\end{align*}

\medskip
\noindent 
where the blocks indexed by $1,\cdots,t$ are omitted when $t=0$, and the last block is omitted when $r=0$.
Then for every prime power $q$, there exists an optimal
$[\mathcal F_n,6,n+3]_q$ Ferrers diagram rank-metric code.
}

In addition, we construct a binary Ferrers diagram rank-metric code directly by a puncturing-inclusion construction and obtain the following result  (see Theorem~\ref{thm:binary}):

\noindent {\bf Theorem C}  {\it Let $m\ge3$, $1\le t\le m-1$ and $$
\mathcal F_{m,t}=[\underbrace{m,m,\ldots,m}_{m\text{ columns}},m+t].$$

\medskip
\noindent Then there  exists an optimal $[\mathcal F_{m,t},m+t+1,m]_2$  Ferrers diagram rank-metric code.}

The paper is organized as follows. Section~\ref{sec:preliminaries} recalls the necessary definitions and basic results needed in the proof of the main result, including Ferrers diagram rank-metric codes, semifields and semifield codes, and the puncturing--inclusion construction. Sections~\ref{sec:odd} and~\ref{sec:even} prove the main results in odd and even characteristic, respectively. As applications of the main results, explicit examples for \(q=7,8,9\) are also provided. The small-field cases \(q=2,3,4\) are not covered by the general proof; for these, the optimal codes are given explicitly. Section \ref{sec:general} constructs optimal codes for two infinite families of Ferrers diagrams. The final section gives the conclusions of the paper.

\section{Preliminaries}\label{sec:preliminaries}

Throughout this paper, we fix a prime power $q$, and set $F=\mathbb F_q,  F^*=F\setminus\{0\}$  and $K=\mathbb F_{q^3}$.
Let $\sigma:K\to K,\ x\mapsto x^q$ be the Frobenius map. We denote the trace map by
$
\operatorname{Tr}_{K/F}:K\to F,
$
which is given by
\[
\operatorname{Tr}_{K/F}(x)=x+\sigma(x)+\sigma^2(x)
\]
for all $x\in K$. Its kernel is denoted by
\[
W=\ker(\operatorname{Tr}_{K/F})
=\{x\in K\mid \operatorname{Tr}_{K/F}(x)=0\}.
\]
The norm map is denoted by
$
N_{K/F}:K^*\to F^*,
$
and is defined by
\[
N_{K/F}(x)=x\cdot\sigma(x)\cdot\sigma^2(x)=x^{1+q+q^2}
\]
for all $x\in K^*=K\setminus\{0\}$.
\subsection{Ferrers diagram rank-metric codes}

Let \(q\) be a prime power. For two matrices \(A,B\in  \mathcal M_{m\times n}(F)\), the rank distance between them is defined by
\[
d_R(A,B)=\operatorname{rank}(A-B).
\]
An \((m\times n,N,\delta)_q\) rank-metric code is a subset
$
\mathcal C\subseteq \mathcal M_{m\times n}(F)
$
of size \(N\) with minimum rank distance
\[
\delta=\min_{A\neq B\in \mathcal C} d_R(A,B).
\]
By the Singleton-like bound,
$
N\le q^{\max\{m,n\}(\min\{m,n\}-\delta+1)}.
$
If equality holds, then \(\mathcal C\) is called a {\bf maximum rank-distance} (MRD) code. If \(\mathcal C\) is a \(k\)-dimensional \(\mathbb F_q\)-linear subspace, it is denoted by \([m\times n,k,\delta]_q\).

Endow $\mathcal{M}_{m\times n}(F)$ with the Frobenius product: 
$$\langle A,B\rangle=\operatorname{tr}(AB^{  T})
=\sum_{i=1}^m\sum_{j=1}^nA_{ij}B_{ij},
$$
Let $\mathcal C$ be a subspace of  $\mathcal{M}_{m\times n}(F)$. Denote
$$
\mathcal{C}^\perp=\{B\in \mathcal{M}_{m\times n}(F)\mid \langle A,B\rangle=0\text{ for all }A\in\mathcal{C}\}, 
$$ which is referred to as the {\bf Delsarte dual} of $\mathcal{C}$.

\begin{proposition}\cite{delsarte1978bilinear}\label{thm:delsarte-duality}
	Let $\mathcal{C}\subseteq \mathcal{M}_{m\times n}(F)$ be a linear MRD code with minimum rank distance $2\le d\le\min\{m,n\}$. Then $\mathcal{C}^\perp$ is again an MRD code. More precisely, if $m\le n$ and
	$
	\mathcal{C}$ is an $[m\times n,n(m-d+1),d]_q
	$ MRD code,
	then
	$
	\mathcal{C}^\perp$ is an $[m\times n,n(d-1),m-d+2]_q
	$ MRD code.
\end{proposition}

	A \textbf{Ferrers diagram} \(\mathcal{F}\) of size \(m \times n\) is an \(m \times n\) array of dots and blanks such that the number of dots in each row is nonincreasing from top to bottom, all dots are right-justified, the first row has \(n\) dots, and the rightmost column has \(m\) dots.
 Also, Ferrers diagram can be described by the sequence of the numbers of dots in its columns. For example, the Ferrers diagram 
$$\mathcal F=\quad \begin{array}{cccc}
	\bullet&\bullet&\bullet&\bullet\\
	\bullet&\bullet&\bullet&\bullet\\
	\bullet&\bullet&\bullet&\bullet\\
	&&&\bullet\\&&&\bullet
\end{array}$$
can be represented by $\mathcal F=[3,3,3,5].$

Note that the column heights in this paper are non-decreasing from left to right, whereas the convention in \cite{BeelooNeri2026} uses a non-increasing column-height order. When applying Theorem~7.5 of \cite{BeelooNeri2026}, we first transpose the diagram, i.e., interchange rows and columns, and then reverse the order of the columns so that it conforms to the convention adopted in this paper.

An \([m\times n,k,\delta]_q\) rank-metric code \(\mathcal C\) is called an \([\mathcal F,k,\delta]_q\) {\bf Ferrers diagram rank-metric code} (FDRMC) if every matrix in \(\mathcal C\) has zero entries outside \(\mathcal F\). The Ferrers diagram rank-metric code of a full Ferrers diagram is exactly a rank-metric code. For given \(\mathcal F\) and \(\delta\), let \(\dim(\mathcal F,\delta)\) denote the largest possible dimension of an \([\mathcal F,k,\delta]_q\) Ferrers diagram rank-metric code. The following Singleton-type bound holds: for \(0\le i\le \delta-1\), let \(\nu_i\) be the number of dots in \(\mathcal F\) after deleting the first \(i\) rows and the rightmost \(\delta-1-i\) columns. Then
\[
\dim(\mathcal F,\delta)\le \min_{0\le i\le \delta-1}\nu_i=:\nu_{\min}(\mathcal F,\delta).
\]
A Ferrers diagram rank-metric code attaining this bound is said to be  optimal. 
It was conjectured by Etzion and Silberstein  that this Singleton-like bound is always tight:

\noindent {\bf Etzion-Silberstein Conjecture}\cite{EtzionSilberstein2009}: {\it For any given Ferrers diagram \(\mathcal F\) and positive integer \(\delta\), there always exists an optimal \([\mathcal F,\nu_{\min}(\mathcal F,\delta),\delta]_q\) Ferrers diagram rank-metric code. }

This conjecture has been verified in several cases, see  \cite{Etzion2016,ZhangGe2019,AntrobusGluesingLuerssen2019,NeriStanojkovski2024,LiuChangFeng2019}. Recently, this conjecture has  been disproved by a counterexample in \cite{Prajapati2026}. However, determining which Ferrers diagrams admit optimal Ferrers diagram rank-metric codes remains an interesting problem worthy of further investigation, due to their crucial role in the multilevel construction of constant dimension codes, please see 
\cite{XuLi2026} and the references therein.

For the Ferrers diagram \(\mathcal{F}=[3,3,3,5]\) with  rank distance \(3\), we have \(\nu_{\min}(\mathcal F,3)=6\).  The paper aims to prove the existence of an optimal \([[3,3,3,5],6,3]_q\) Ferrers diagram rank-metric code. This was posed as an open problem in \cite{Etzion2016}, where a binary example was also given. Here we provide a general  construction for every prime power $q$.

\subsection{Semifield Codes}

This subsection introduces a class of MRD codes induced by semifields. 
Let $F=\mathbb F_q$, and let $\mathbb K/F$ be a finite field
extension. An {$F$-semifield structure} on $\mathbb K$
is an $F$-bilinear multiplication
$
\circ:\mathbb K\times\mathbb K\longrightarrow\mathbb K
$
such that: (1) for all $x,y\in\mathbb K$, $x\circ y=0$ implies 
$x=0$ or $y=0$; and (2) there exists $e\in\mathbb K\setminus\{0\}$ 
satisfying $e\circ x=x\circ e=x$ for every $x\in\mathbb K$.
Equipped with its usual addition and the multiplication $\circ$,
the structure $(\mathbb K,+,\circ)$ is called an $F$-{\bf semifield}.

We introduce the  construction of  semifields by  the classical Kaplansky's trick. This technique has been systematically discussed in Albert's twisted field construction and its generalizations; see \cite{albert1961a, pumplun2015}. The construction  in this subsection can be viewed as a special case of Kaplansky's trick. For the reader's convenience, we provide detailed proofs.

\begin{proposition}\emph{\cite{albert1961a}}\label{bilinear1}
	Let \(c \in F^*\) satisfy \(c^3 \neq -1\). Define an operation \(\star\) on \(K\):
	\[
	a \star b = ab + c a^q b^{q^2}, \quad \forall a,b \in K.
	\]
	Then \(\star\) is an \(F\)-bilinear map without zero divisors.
\end{proposition}

\begin{proof}
Since the Frobenius maps \(x\mapsto x^q\) and \(x\mapsto x^{q^2}\)
are \(F\)-linear, the operation \(\star\) is \(F\)-bilinear.
Suppose that \(a,b\in K^*\) and \(a\star b=0\). Dividing
\(ab+ca^qb^{q^2}=0\) by \(cab\), we obtain
\[
a^{q-1}b^{q^2-1}=-c^{-1}.
\]
Since \(N_{K/F}(a),N_{K/F}(b)\in F^*\), taking norms gives
\[
1=N_{K/F}(a)^{q-1}N_{K/F}(b)^{q^2-1}
 =N_{K/F}(-c^{-1})=(-c^{-1})^3=-c^{-3}.
\]
Thus \(c^3=-1\), contrary to the hypothesis. Therefore,
\(\star\) has no zero divisors.
\end{proof}

\begin{proposition}\emph{\cite{pumplun2015}}\label{bilinear2}
	Let \(c \in F^*\) satisfy \(c^3 \neq -1\), and let the operation $\star$ be as defined in Proposition \emph{\ref{bilinear1}}. Define the  right and left multiplication maps \(R, L: K \to K\):
	\[
	R(a) = a \star 1, \quad L(b) = 1 \star b.
	\]
	Then \(R, L\) are both invertible linear maps on \(K\). Furthermore, let $e=1+c$,
	and define an operation \(\circ\):
	\[
	a \circ b = e^{-1} \left( R^{-1}(ea) \star L^{-1}(eb) \right), \quad \forall a,b \in K.
	\]
	Then \((K, \circ)\) is a semifield.
\end{proposition}

\begin{proof}
	Since \(\star\) has no zero divisors, \(R\) and \(L\) are injective \(F\)-linear maps on the finite-dimensional space \(K\), hence invertible. Moreover, \(c^3\ne-1\) implies \(e=1+c\ne0\). These facts ensure that \(\circ\) is well defined and has no zero divisors.
	Since $R^{-1}$ and $L^{-1}$ are both linear maps, and $\star$ is a bilinear map, one can see that $\circ $ is a bilinear map.
	
	Since 
	$
	L(1)=1\star 1=e,
	$
	we have $L^{-1}(e)=1$. Hence
	$$
	a\circ 1
	=e^{-1}\left(R^{-1}(ea)\star L^{-1}(e)\right)
	=e^{-1}\left(R^{-1}(ea)\star 1\right)=e^{-1}R\left(R^{-1}(ea)\right)=e^{-1}ea=a.
	$$
	Similarly, one can prove that $1\circ a=a.$
	Therefore, $1$ is the identity element of $(K,\circ)$.
	In summary, $(K,\circ)$ is a semifield.
\end{proof}

\begin{remark}
When $c=0$, both $\star$ and $\circ$ reduce to ordinary field multiplication on $K$, so the commutator vanishes and this case does not yield the nonzero commutator required in the construction below.
\end{remark}

In general, for arbitrary \(u,v\), computing \(u\circ v\) can be complicated. For our later purposes, however, we need the commutator \(u\circ v-v\circ u\), and this simplifies when \(u\) and \(v\) are trace-zero.

\begin{proposition}\label{prop:trace-zero-master} 
	Keep the notation of Proposition~\ref{bilinear2}, and set $D=c^2-c+1$. Let $u,v\in W$.
	Then we have the following.
	
	\begin{itemize}

		\item[$(1)$] $R^{-1}(eu)
		=
		\frac{e}{D}\bigl((1-c)u-cu^q\bigr),
		\
		L^{-1}(ev)
		=
		\frac{e}{D}\bigl(v+cv^q\bigr).$
		\item[$(2)$]  $u\circ v-v\circ u
		=
		\frac{3c(1-c^2)}{(c^2-c+1)^2}
		\bigl(uv^q-u^qv\bigr).$
	\end{itemize}
\end{proposition}

\begin{proof} Set $D=c^2-c+1$. Since $c^3\neq -1$ and  $c^3+1=(c+1)(c^2-c+1)$, we  have $D\neq 0.$

	 (1)
	For $w\in W$,
	$
	w+w^q+w^{q^2}=0.
	$
	Writing $\sigma(w)=w^q$, we therefore have on $W$
	\begin{equation}\label{eq:master-frob}
		\sigma^2+\sigma+I=0.
	\end{equation}
Note that 
	$
	R=I+c\sigma,
	$ and $
	L=I+c\sigma^2=(1-c)I-c\sigma
	$ holds on  $W$. 
	It follows from  \eqref{eq:master-frob} that 
	\[
	RL= (I+c\sigma)\bigl((1-c)I-c\sigma\bigr)=DI
	\]
	holds on $W$.
	Hence, we have on $W$ that 
	$
	R^{-1}=\frac1D\bigl((1-c)I-c\sigma\bigr)$ 
	and
	$
	L^{-1}=\frac1D(I+c\sigma).
	$
	This  proves (1).
	
	(2) Set $X=(1-c)u-cu^q,Y=v+cv^q.$
	By bilinearity, we get
	$
	u\circ v=\frac{e}{D^2}(X\star Y).
	$
	Since
	$$
	u^{q^2}=-u-u^q,
	\quad
	v^{q^2}=-v-v^q,
	$$
	we obtain
	$
	X^q=cu+u^q,
	\
	Y^{q^2}=(c-1)v-v^q.
	$
	Expanding the two terms separately gives
	\begin{align*}
		XY
		={}&(1-c)uv+c(1-c)uv^q-cu^qv-c^2u^qv^q,\\
		cX^qY^{q^2}
		={}&c^2(c-1)uv-c^2uv^q
		+c(c-1)u^qv-cu^qv^q.
	\end{align*}
	Therefore,
	\[ \begin{aligned}
		u\circ v
		&=\frac{e}{D^2}
		(X\star Y)=\frac{e}{D^2}(XY+cX^qY^{q^2})
		\\
		&=	\frac{1+c}{(c^2-c+1)^2}
		\Bigl[
		(1+c)(1-c)^2uv
		+(c-2c^2)uv^q
		+(c^2-2c)u^qv
		-c(1+c)u^qv^q
		\Bigr].
	\end{aligned}\]
	Now interchange $u$ and $v$. Since the original field multiplication on $K$ is commutative, the coefficients of $uv$ and $u^qv^q$ are unchanged, while the coefficients of $uv^q$ and $u^qv$ are exchanged. Hence
	\begin{align*}
		u\circ v-v\circ u
		&=
		\frac{1+c}{D^2}
		\bigl[(c-2c^2)-(c^2-2c)\bigr]
		(uv^q-u^qv)\\
		&=
		\frac{1+c}{D^2}\,3c(1-c)(uv^q-u^qv)=
		\frac{3c(1-c^2)}{D^2}(uv^q-u^qv).
	\end{align*}
	This proves (2).
\end{proof}

The following result is an easy consequence of elementary finite field theory.
\begin{lemma}\label{lem:trace-zero-master}  Let $u,v\in K^*$ and set $\Omega(u,v)=uv^q-u^qv$. Then we have the following.
	\begin{itemize}
		\item[(1)] If $\Omega(u,v)=uv^q-u^qv=0$, then $(u,v)$ is linearly dependent  over $F$.
		\item[(2)] If $u,v\in W$, then 
		$\Omega(u,v)\in F.$ Moreover, if  $(u,v)$ is linearly independent  over $F$, then 	$\Omega(u,v)\in F^*.$ 
	\end{itemize}
	
\end{lemma}

\begin{proof} We give a detailed proof for the reader's convenience.
	
	(1) Suppose  $\Omega(u,v)=0$. Then 
	$
	uv^q=u^qv.
	$
	Since  $u,v$ are nonzero, divide by $vv^q$ to obtain
	\[
	\frac{u}{v}=\left(\frac{u}{v}\right)^q.
	\]
	Hence $u/v\in F$, which implies that  $(u,v)$ is linearly dependent over $F$.

	(2)	For $u,v\in W$,  we have
	$
	u^{q^2}=-u-u^q,
	v^{q^2}=-v-v^q.
	$
	Hence, step by step,
	\begin{align*}
		\Omega(u,v)^q
		&	=(uv^q-u^qv)^q
		=u^qv^{q^2}-u^{q^2}v^q\\
		&=u^q(-v-v^q)-(-u-u^q)v^q
		=-u^qv-u^qv^q+uv^q+u^qv^q\\
		&=uv^q-u^qv
		=\Omega(u,v).
	\end{align*}
	Thus $\Omega(u,v)$ is fixed by the $q$-Frobenius and therefore
	$
	\Omega(u,v)\in F.
	$ Moreover, if  $(u,v)$ is linearly independent  over $F$, then  it follows from (1) that	$\Omega(u,v)\in F^*.$ 
	\end{proof}

At the end of this subsection,  we   give the definition of a semifield code in general.
\begin{lemma}\emph{\cite{deLaCruz2016}}\label{lem:trace-spread}
	Let $(\mathbb K,\circ)$ be an $ F$-semifield, and set $n=[\mathbb K: F]$. Take two ordered $ F$-bases of $\mathbb K:$
	\[
	\mathbf p=(p_1,\cdots,p_n),\qquad \mathbf q=(q_1,\cdots,q_n),
	\]
	For $z\in \mathbb K$, define
	$$
	B_z=\bigl(\mathrm{Tr}_{\mathbb K/ F}(z(p_i\circ q_j))\bigr)_{1\le i,j\le n}.$$
	Then $
	\mathcal S=\{B_z\mid z\in \mathbb K\}$
	is an $[n\times n,n,n]_q$ MRD code, called a {\bf semifield code}.
\end{lemma}

\subsection{Puncturing-inclusion constructions}

The puncturing--inclusion construction was introduced by Beeloo--Sauerbier Couvée and Neri in \cite{BeelooNeri2026} to study the Etzion--Silberstein conjecture. The main idea is to perform puncturing (deleting a row) and inclusion (embedding into a larger MRD code) on rectangular MRD codes, thereby constructing Ferrers diagram rank metric  codes. 

\begin{theorem}\cite[Theorem 7.5]{BeelooNeri2026}  \label{thm:d3-mrd-correspondence}
	Let $n,d\in\mathbb N$ with $1<d<n$ and  $$\mathcal{F}=[\underset{n-1}{\underbrace{n-1,n-1,\cdots,n-1}},n+d-2].$$ Then the following are equivalent.
	\begin{enumerate}
		\item[(1)] There exists an  optimal $[\mathcal F, \nu_{\min}(\mathcal F,d), d]_q$ Ferrers diagram rank-metric code.
		\item[(2)] [Puncturing--inclusion MRD condition] \ There exists an $[(n-1)\times n,n(n-d),d]_{q}$ MRD code $\mathcal{C}$ such that its puncturing on the last column is contained in an $[(n-1)\times(n-1),(n-1)(n-d+1),d-1]_{q}$ MRD code.
	\end{enumerate}
\end{theorem}

The paper aims to prove    the puncturing--inclusion MRD  condition for $(n,d)=(4,3)$, the following proposition serves as the starting point of our study.

\begin{proposition}\label{prop:mrd-ferrers-lift}
	Let $\mathcal L$ be a subspace of $\mathcal{M}_3(F)$ with $\dim_F \mathcal L=4$. Suppose  that there exists an $F$-linear map $w:\mathcal L\to F^3$ such that
	$$
	\mathcal{K}=\Big\{\left[B\mid w(B)\right]\mid B\in \mathcal L\Big\}
	$$
	is a $[3\times4,4,3]_q$ MRD code. If there exists a $[3\times3,3,3]_q$ MRD code
	$
	\mathcal{S}\subseteq \mathcal L^\perp,
	$
	then there exists an optimal
	$
	[[3,3,3,5],6,3]_q
	$
	Ferrers diagram rank-metric code.
\end{proposition}

\begin{proof}
	Suppose that there exists a $[3\times3,3,3]_q$ MRD code
	$
	\mathcal{S}\subseteq \mathcal L^\perp,
	$
	then $\mathcal L\subseteq \mathcal{S}^\perp$. By Proposition \ref{thm:delsarte-duality}, $\mathcal{S}^\perp$ is a $[3\times3,6,2]_q$ MRD code. Thus  the puncturing--inclusion MRD condition holds for $(n,d)=(4,3)$. By Theorem \ref{thm:d3-mrd-correspondence}, there exists an optimal
	$
	[[3,3,3,5],6,3]_q
	$
	Ferrers diagram rank-metric code.
\end{proof}

\begin{remark}\label{rem:mrd-ferrers-lift-correspondence}
	We describe how to obtain an optimal Ferrers diagram rank-metric code from the MRD code $\mathcal S$
	in Proposition \ref{prop:mrd-ferrers-lift}. This process is precisely the puncturing--inclusion construction introduced in \cite{BeelooNeri2026}.
	Let $\{B_i\mid 1\leq i\leq 6\}$ be a basis of $\mathcal S^\perp$ such that $\{B_1,B_2,B_3,B_4\}$ is a basis of $\mathcal L$. Define
	$$
	K_i=\begin{pmatrix}
		B_i & w(B_i)\\
		O_{1\times 3}&0\\
		O_{1\times 3} &0
	\end{pmatrix}
	(1\leq i\leq 4),
\quad 
	K_5=\begin{pmatrix}
		B_5 & O_{3\times 1}\\
		O_{1\times 3}&1\\
		O_{1\times 3} &0
	\end{pmatrix},
	\quad
	K_6=\begin{pmatrix}
		B_6 & O_{3\times 1}\\
		O_{1\times 3}&0\\
		O_{1\times 3} &1
	\end{pmatrix},
	$$
	Then
	$$
	\mathcal C=\operatorname{span}\{K_i\mid 1\leq i\leq 6\}=\left\{\Gamma(B,a,b)\mid B\in \mathcal L, (a,b)\in F^2\right\}
	$$
	is the desired optimal Ferrers diagram rank-metric code, where
	$$\Gamma(B,a,b)=\begin{pmatrix}
		B+aB_5+bB_6&w(B)\\
		O_{1\times 3}&a\\
		O_{1\times 3} &b
	\end{pmatrix}$$
\end{remark}

In the sequel, we will construct optimal $[[3,3,3,5],6,3]_q$ Ferrers diagram rank-metric codes based on Proposition \ref{prop:mrd-ferrers-lift}. Specifically, we will first construct a $[3\times4,4,3]_q$ MRD code, and then prove that the Delsarte dual of its puncturing on the  last  column contains a $[3\times3,3,3]_q$ MRD code, and this contained three-dimensional MRD code is chosen from semifield codes.

\section{Proof of the main results when \texorpdfstring{$q$}{q} is odd}\label{sec:odd}

In this section, we prove that there exists an optimal \([[3,3,3,5],6,3]_q\) Ferrers diagram rank-metric code for $q$ odd.
 We begin our proof with the construction of MRD codes of full row rank. 
 Throughout this section, $q$ is assumed to be odd.
 
\subsection{Construction of full row rank MRDs}

Let
$
M:F^{m+1}\rightarrow \mathcal M_{m\times(m+1)}(F)
$
be an \(F\)-linear map, and write
$
X=(x_1,x_2,\dots,x_{m+1})^{\mathsf T}.
$
For \(u=(a_1,\dots,a_m)\in F^m\), the bilinear map
\((u,X)\longmapsto uM(X)\) uniquely determines an \((m+1)\times(m+1)\) matrix \(H_M(u)\) such that
\begin{equation}\label{eq:left-kernel-identity}
	uM(X)=\big(H_M(u)X\big)^T.
\end{equation}

\begin{proposition}\label{prop:full-row-rank}
	Assume that \(M\) is injective. If \(H_M(u)\) is invertible for all nonzero vectors \(u\in F^m\), then
	\[
	\mathcal C=\{M(X)\mid X\in F^{m+1}\}
	\]
	is an \([m\times(m+1),m+1,m]_q\) MRD code.
\end{proposition}

\begin{proof}
	Let \(u\in F^m\) be a nonzero vector. Suppose that \(uM(X)=0\). Then \eqref{eq:left-kernel-identity} gives \(H_M(u)X=0\).  By assumption, \(H_M(u)\) is invertible. It follows that \(X=0\). Since \(M\) is injective, we get \(M(X)=0\). Consequently, each nonzero matrix \(M(X)\) is of full row rank, namely, rank \(M(X)=m\). This proves that \(\mathcal C\) is an \([m\times(m+1),m+1,m]_q\) MRD code.
\end{proof}

According to  Proposition \ref{prop:full-row-rank},  one can construct a \([3\times 4,4,3]_q\) MRD code by finding  a determinant in three variables that is invertible for every nonzero choice of the variables. A natural approach is to exploit special properties of field elements, such as nonsquare elements or trace-zero elements. Therefore, as long as we  find a polynomial in three variables with no nonzero zero over $F$ and construct the corresponding determinant, we can obtain an MRD code satisfying the required conditions. One candidate polynomial  is
\[
f(x,y,z)=\Delta^3x^4-\bigl[z^2+\Delta(tx^2-y^2)\bigr]^2,
\]  
where \(t\in F\) and \(\Delta\) is a nonsquare element.
In fact, suppose that  $f(a,b,c)=0$ for $a,b,c\in F$. 
If \(a\ne0\), then
\[
\Delta=\left(\frac{c^2+\Delta(ta^2-b^2)}{\Delta a^2}\right)^2,
\]
which contradicts the fact that \(\Delta\) is nonsquare.  It follows that  \(a=0\) and thus  \(c^2=\Delta b^2\). Now  the nonsquare property  of  $\Delta$ forces \(b=c=0\). Consequently,  for every nonzero \((a,b,c)\), $f(a,b,c)\neq 0.$

 Based on the polynomial $f(x,y,z)$, we present the corresponding determinant:
$$ \det H_{\Delta,t}(a,b,c)=	\begin{vmatrix}
	a&0&b&c\\
	b&c&at&a\Delta\\
	c&\Delta b&-\Delta a&-t\Delta a\\
	0&\Delta a&-c&-\Delta b
\end{vmatrix}=\Delta^3a^4-\bigl[c^2+\Delta(ta^2-b^2)\bigr]^2.$$
Let $(R,S,C,E)\in F^4$. By \eqref{eq:left-kernel-identity}, we have 
	$$ (a,b,c)\begin{pmatrix}
		R&tC+\Delta E&-\Delta C-t\Delta E&\Delta S\\
		C&R&\Delta S&-\Delta E\\
		E&S&R&-C
	\end{pmatrix}=	\left(\begin{pmatrix}
	a&0&b&c\\
	b&c&at&a\Delta\\
	c&\Delta b&-\Delta a&-t\Delta a\\
	0&\Delta a&-c&-\Delta b
\end{pmatrix}\begin{pmatrix}
	R\\S\\C\\E
\end{pmatrix}\right)^T.$$

Since  for every nonzero \((a,b,c)\), $f(a,b,c)\neq 0$, we get that  \(H_{\Delta,t}(a,b,c)\) is invertible. As a direct consequence of Proposition \ref{prop:full-row-rank}, we get the following theorem.
\begin{theorem}\label{thm:odd-core}
	Let  \(\Delta\in F\)  be a nonsquare element and  \(t\in F\). For
	$(R,S,C,E)\in F^4$, set 
	\begin{equation*}\label{eq:odd-core}
		M_{\Delta,t}(R,S,C,E)=
		\begin{pmatrix}
			R&tC+\Delta E&-\Delta C-t\Delta E&\Delta S\\
			C&R&\Delta S&-\Delta E\\
			E&S&R&-C
		\end{pmatrix}.
	\end{equation*}
	Then
	\[
	\mathcal{K}_{\Delta,t}=\{M_{\Delta,t}(R,S,C,E)\mid R,S,C,E\in F\}
	\]
	is a \([3\times4,4,3]_q\) MRD code.
\end{theorem}

\subsection{The Orthogonality Conditions}

To perform the puncturing-inclusion construction, we write \(M_{\Delta,t}=[A_{\Delta,t}\mid w]\), where
\begin{equation*}\label{eq:odd-Aw}
	A_{\Delta,t}=\begin{pmatrix}
		R&tC+\Delta E&-\Delta C-t\Delta E\\
		C&R&\Delta S\\ E&S&R
	\end{pmatrix},\qquad
	w=\begin{pmatrix}\Delta S\\-\Delta E\\-C\end{pmatrix}.
\end{equation*}
Set  \(\mathcal L_{\Delta,t}=\{	A_{\Delta,t}\mid R,S,C,E\in F\}\). Based on Proposition \ref{prop:mrd-ferrers-lift}, we proceed to prove that $\mathcal L_{\Delta,t}^\perp$ contains a $[3\times3,3,3]_q$
MRD code and 
we find such a code among semifield codes over \(K\).
  
Consider the  cubic extension \(K/F\) and let \((K,\circ)\) be a semifield.
Let \(\mathbf p=(p_1,p_2,p_3)\) and \(\mathbf q=(q_1,q_2,q_3)\) be two bases of $K$ over $F$.  Set 
\[
B_z=\Bigl(\Tr_{K/F}\big(z(p_i\circ q_j)\big)\Bigr)_{i,j},\qquad
\mathcal S=\{B_z:z\in K\}.
\]
It follows from Lemma \ref{lem:trace-spread}  that  $\mathcal S$ is a $[3\times3,3,3]_q$
MRD code.
Write
\[
P_{ij}=p_i\circ q_j.
\]
From \eqref{eq:odd-Aw} and the definition of \(B_z\), pairing term by term gives
\[
\begin{aligned}
	\langle A_{\Delta,t},B_z\rangle
	=\Tr_{K/F}\bigl(z\{&
	R(P_{11}+P_{22}+P_{33})+S(\Delta P_{23}+P_{32})\\
	&+C(tP_{12}-\Delta P_{13}+P_{21})+E(\Delta P_{12}-t\Delta P_{13}+P_{31})\}\bigr).
\end{aligned}
\]
Since \(R,S,C,E\) range freely over \(F\) and \(z\) ranges freely over \(K\), by the nondegeneracy of the trace pairing,
we get that $
\mathcal S\subseteq \mathcal L_{\Delta,t}^{\perp}
$
if and only if the four coefficients vanish in \(K\), that is,
\begin{align}
	p_1\circ q_1+p_2\circ q_2+p_3\circ q_3&=0,\label{eq:O1}\\
	\Delta p_2\circ q_3+p_3\circ q_2&=0,\label{eq:O2}\\
	tp_1\circ q_2-\Delta p_1\circ q_3+p_2\circ q_1&=0,\label{eq:O3}\\
	\Delta p_1\circ q_2-t\Delta p_1\circ q_3+p_3\circ q_1&=0.\label{eq:O4}
\end{align}

 Set 
	$
	\mathbf p=(1,x,y),\ q_1=1,\ \delta=\Delta-t^2.
	$
 Since $\Delta$ is nonsquare element, we have  \(\delta\ne0\) and \eqref{eq:O3}--\eqref{eq:O4} uniquely determine
	\begin{equation}\label{eq:odd-r23}
		q_2=\frac{tx-y}{\delta},\qquad
		q_3=\frac{\Delta x-ty}{\Delta\delta}.
	\end{equation}
	After substituting back, \eqref{eq:O1}--\eqref{eq:O2} are equivalent to
	\begin{align}
		x\circ y-y\circ x&=\Delta,\label{eq:odd-target1}\\
		y\circ y-\Delta(x\circ x)&=-t\Delta.\label{eq:odd-target2}
	\end{align}

\begin{proposition}\label{prop-keyequ} Let $(K,\circ)$ be an $F$-semifield and fix a nonsquare $\Delta\in F$.
	 If there exists a basis $(1,u,v)$ of $K$ over $F$  such that
	\begin{equation*}
	u\circ v-v\circ u\in F^*,  \label{keyequ}
\end{equation*}
	then there exist $t\in F$ and a semifield code $\mathcal S$ such that $
	\mathcal S\subseteq \mathcal L_{\Delta,t}^{\perp}
	$.
\end{proposition}

\begin{proof} Based on the above statements, it suffices to show that there exists  a basis $(1,x,y)$ of $K$ over $F$  satisfying  (\ref{eq:odd-target1}) and (\ref{eq:odd-target2}).
	Let  $(1,u,v)$  be a basis of $K$ over $F$  such that
$$
		u\circ v-v\circ u=\lambda\in F^*.
$$
Set $v_0=\frac{\Delta}{\lambda}v$. By the 	bilinearity of $\circ$, we get that 
$u\circ  v_0-v_0\circ u=\Delta.$
Set
$
x=u+a,
y=v_0+b,
$ we seek $a,b\in F$ such that
\[
y\circ y-\Delta(x\circ x)=-t\Delta.
\]
Since  $1$ is a two-sided identity and $a,b\in F$, we have 
\begin{align*}
	y\circ y
	&=(v_0+b)\circ(v_0+b)\\
	&=v_0\circ v_0+b(v_0\circ1)+b(1\circ v_0)+b^2(1\circ1)\\
	&=v_0\circ v_0+2bv_0+b^2.
\end{align*}
Similarly,
$
x\circ x=u\circ u+2au+a^2.
$
Therefore
\begin{align*}
	y\circ y-\Delta(x\circ x)
	={}&v_0\circ v_0-\Delta(u\circ u)\\
	&+2bv_0-2\Delta au+(b^2-\Delta a^2).
\end{align*}
Suppose $v_0\circ v_0-\Delta(u\circ u)=s_1+s_2u+s_3v_0$ with $(s_1,s_2,s_3)\in F^3$.
Then $y\circ y-\Delta(x\circ x)=-t\Delta$ if and only if 
$$\begin{array}{ll}&s_1+s_2u+s_3v_0+2bv_0-2\Delta au+(b^2-\Delta a^2)\\
=&(s_1+b^2-\Delta a^2)+(s_2-2\Delta  a) u+(s_3+2b)v_0
\\
=&-t\Delta.
\end{array}$$
Set $a=\frac{s_2}{2\Delta}, b=-\frac{s_3}{2}, t=-\frac{s_1+b^2-\Delta a^2}{\Delta}$. 
Then we get that $$y\circ y-\Delta(x\circ x)=-t\Delta.$$
By the bilinearity of $\circ$, one can check that $$x\circ y-y\circ x=\Delta.$$ Finally, since $(1,u,v)$ is an $F$-basis of $K$, one can check that 
$(1,x,y)$ is also an $F$-basis of $K$. Moreover, the linear transformation from $(x,y)$ to $(q_2,q_3)$ in~\eqref{eq:odd-r23} has determinant $1/(\Delta\delta)\ne0$. Hence $(1,q_2,q_3)$ is an $F$-basis of $K$ as required. This completes the proof of the proposition.
\end{proof}

\begin{remark}
By Proposition \ref{prop-keyequ}, to construct a semifield code \(\mathcal{S}\) such that
$
\mathcal{S}\subseteq \mathcal L_{\Delta,t}^{\perp},
$
the key point is to choose a suitable basis \((1,u,v)\) of $K$ over $F$ satisfying $u\circ v-v\circ u\in F^*$.
\end{remark}

\subsection{Proof of the main result for \texorpdfstring{$\operatorname{char}(F)\notin\{2,3\}$}{char(F) not in (2,3)}}

In this subsection, we suppose that char$(F)$ is odd and char$(F)\neq 3$.  Then there exists $c\in F^*\setminus\{1,-1\}$ such that $c^3\neq -1$. Since    char$(F)\neq 3$, we get $ \Tr_{K/F}(1)=3\neq 0$. It follows that $K=F\oplus W.$ Recall that $W=\mathrm{ker}\ ( \Tr_{K/F}).$
Let $(u,v)$ be an  $F$-basis of $W$. Then $(1,u,v) $ is an $F$-basis of $K$.
It follows from Proposition  \ref{prop:trace-zero-master}  and Lemma \ref{lem:trace-zero-master}  that 
$$u\circ v-v\circ u
 =\textstyle{\frac{3c(1-c^2)}{(c^2-c+1)^2}}
 \bigl(uv^q-u^qv\bigr)\in F^*.$$
Then we get  from Proposition \ref{prop-keyequ} that there exists  a semifield code \(\mathcal{S}\) such that
$
\mathcal{S}\subseteq \mathcal L_{\Delta,t}^{\perp}.
$
 Finally, it follows  from Proposition \ref{prop:mrd-ferrers-lift} that 
  there exists an optimal
 $
 [[3,3,3,5],6,3]_q
 $
 Ferrers diagram rank-metric code. 
 Consequently, we have proved  the following theorem.

 \begin{theorem}\label{thm:charno3-general}
 	Suppose that $q$ is odd and char $\mathbb F_q\neq 3.$ Then there exists an optimal
 	$$
 	[[3,3,3,5],6,3]_q
 	$$
 	Ferrers diagram rank-metric code. 
 \end{theorem}
 
 We now illustrate Theorem \ref{thm:charno3-general} with an example.
 \begin{example}[An explicit code over $\F_{7}$]\label{ex:q7}
 	Let $F=\F_7$ and $K=F(\beta)$, where $\beta^3+2=0$. The displayed cubic has no root in $F$, so it is irreducible. Choose
 	\[c=2,\qquad u=\beta,\qquad v=\beta^2,\qquad \lambda=u\circ v-v\circ u=6.\]
 	Both $u$ and $v$ have trace zero, and $(1,u,v)$ is an $F$-basis of $K$.
 	Take the nonsquare $\Delta=3$. The normalization gives
 	\[v_0={4}\beta^2,\qquad v_0\circ v_0-\Delta(u\circ u)={2}u+{3}v_0.\]
 	Consequently, we obtain
 	\[\begin{aligned}x&=5+\beta,&y&=2+{4}\beta^2,\\t&=5,&\delta&=6,\\q_2&=5+{2}\beta+{4}\beta^2,&q_3&=3+{6}\beta+{2}\beta^2.\end{aligned}\]
 	Then we get  two desired bases:
 	$$\mathbf p=(1,5+\beta,2+4\beta^2),\quad \mathbf q=(1,5+2\beta+4\beta^2,3+6\beta+2\beta^2)$$
 The trace calculation yields
 \[
 B_1=
 \begin{pmatrix}
 	3 & 1 & 2 \\
 	1 & 6 & 0 \\
 	6 & 0 & 5
 \end{pmatrix},
 \qquad
 B_\beta=
 \begin{pmatrix}
 	0 & 4 & 2 \\
 	0 & 5 & 0 \\
 	4 & 0 & 2
 \end{pmatrix},
 \qquad
 B_{\beta^2}=
 \begin{pmatrix}
 	0 & 2 & 6 \\
 	1 & 1 & 5 \\
 	0 & 6 & 6
 \end{pmatrix}.
 \]
 A complement of $\mathcal L$ in $\mathcal S^\bot$ is supplied by
 \[H_1=\begin{pmatrix}2&0&1\\0&1&0\\0&0&0\end{pmatrix},\qquad H_2=\begin{pmatrix}4&4&6\\0&0&1\\0&0&0\end{pmatrix}.\]
 	The resulting code is $\mathcal C_{7}=\{N_{7}(R,S,C,E,a,b)\mid R,S,C,E,a,b\in F\}$, where
 	\begin{equation*}\label{eq:example-code-7}N_{7}(R,S,C,E,a,b)=\begin{pmatrix}R+{2}a+{4}b&{5}C+{3}E+{4}b&{4}C+{6}E+a+{6}b&{3}S\\C&R+a&{3}S+b&{4}E\\E&S&R&{6}C\\0&0&0&a\\0&0&0&b
 		\end{pmatrix}.\end{equation*}
 	A SageMath computation shows that the rank distribution of $\mathcal C_7$ is
 	\[
 	(A_0,A_1,A_2,A_3,A_4)=(1,0,0,21552,96096),
 	\]
 	where $A_r=\#\{M\in\mathcal C_7:\operatorname{rank}M=r\}$.
 	
 \end{example}

\subsection{Proof of the main result for \texorpdfstring{$\operatorname{char}(F)=3$}{char(F)=3}}

In this subsection, assume $q=3^m$ with $m\ge2$, so $\operatorname{char}(K)=3$.
In this case,
$
\Tr_{K/F}(1)=3=0,
$
so $1\in W$. Since $\dim_F W=2$, if $u,v\in W$, then $1,u,v$ all lie in the same two-dimensional subspace and cannot form a basis of $K/F$. Consequently, the strategy employed in the case \(\operatorname{char}(F)\neq 3\) is no longer applicable.

Choose $c\in F^*\setminus\{1,-1\}$ such that $c^3\neq -1$, which allows us to define a semifield $(K,\circ)$ as in Proposition \ref{bilinear2}. Set $e=1+c, r=\frac{c}{e}$.
 To find an $F$-basis \((1,u,v)\) of \(K\) such that
$
u\circ v-v\circ u=\lambda\in F^*,
$
let  $I$ be the identity map on $K$. Since $\sigma^3=I, $  we have  $(\sigma-I)^3=0$ and $(\sigma-I)^2\neq 0$. In characteristic $3$, $(\sigma-I)^2=I+\sigma+\sigma^2=\Tr_{K/F}$. The trace map is surjective, so choose $v\in K$ such that $(\sigma-I)^2(v)=1$, and set $u=(\sigma-I)(v)$. To verify linear independence, apply $(\sigma-I)^2$ and then $\sigma-I$ to a relation $a+bu+dv=0$ with $a,b,d\in F$: this gives $d=b=a=0$. Thus $(1,u,v)$ is an $F$-basis of $K$, and we have
\begin{equation}\label{eq:uv-frob}
	u^q=u+1,\qquad v^q=v+u,\qquad
	u^{q^2}=u-1,\qquad v^{q^2}=v-u+1.
\end{equation}
Since 
$
\sigma^3=I,
R=I+c\sigma,
L=I+c\sigma^2,
$ we have 
\[
(I+c\sigma)(I-c\sigma+c^2\sigma^2)
=(1+c^3)I=e^3I,
\]
so
\begin{equation}\label{eq:Rinv}
	R^{-1}(ea)
	=\frac{a-ca^q+c^2a^{q^2}}{e^2}.
\end{equation}
Similarly,
\begin{equation}\label{eq:Linv}
	L^{-1}(ea)
	=\frac{a-ca^{q^2}+c^2a^q}{e^2}.
\end{equation}
From \eqref{eq:uv-frob}--\eqref{eq:Linv}, we obtain
\begin{equation}\label{eq:Linvlinv}
	R^{-1}(eu)=u-r,
	\quad
	L^{-1}(eu)=u+r,
\quad
	R^{-1}(ev)=v-ru+r^2,
	\quad
	L^{-1}(ev)=v+ru-r+r^2.
\end{equation}
Substituting~\eqref{eq:Linvlinv} into $u\circ v=e^{-1}\bigl(R^{-1}(eu)\star L^{-1}(ev)\bigr)$ and simplifying, we obtain
\[
u\circ v-v\circ u=\kappa:=\frac{c(1-c)}{(1+c)^3}\in F.
\]
Since $c\in F^*\setminus\{1,-1\}$,  we have $u\circ v-v\circ u\in F^*.$

Now  we get  from Proposition \ref{prop-keyequ} that there exists  a semifield code \(\mathcal{S}\) such that
$
\mathcal{S}\subseteq \mathcal L_{\Delta,t}^{\perp}.
$
Finally, it follows  from Proposition \ref{prop:mrd-ferrers-lift} that 
there exists an optimal
$
[[3,3,3,5],6,3]_{3^m}
$
Ferrers diagram rank-metric code. 
Consequently, we have proved  the following theorem. 
\begin{theorem}\label{thm:char3-general}
	For every $m\ge2$, there exists an optimal
	\[
	[[3,3,3,5],6,3]_{3^m}
	\]
	Ferrers diagram rank-metric code.
\end{theorem}

We now illustrate Theorem \ref{thm:char3-general} with an example.
\begin{example}[An explicit code over $\F_{9}$]\label{ex:q9}
Let $F=\F_3(\alpha)$ with $\alpha^2+1=0$, and let $K=F(\beta)$ with $\beta^3+\beta+\alpha=0$. The displayed cubic has no root in $F$, so it is irreducible. Choose
\[c=\alpha,\qquad u={\alpha}\beta,\qquad v={\alpha}\beta+\beta^2,\qquad \lambda=u\circ v-v\circ u=\alpha.\]
Here $u^9-u=1$ and $v=u-u^2$ satisfies $v^9-v=u$. Thus $(1,u,v)$ forms an $F$-basis of $K$, as required in characteristic $3$. 
Take the nonsquare $\Delta=\alpha+1$.  Following the procedure in Example \ref{ex:q7}, we obtain two desired bases:
$$\begin{array}{lcl}\mathbf p&=&\big(1,2\alpha+\alpha\beta, \alpha+(\alpha+1)\beta+(2\alpha+1)\beta^2\big)\\
	\mathbf q&=&\big(1,\alpha+(2\alpha+1)\beta+\beta^2, 2\alpha+(2\alpha+1)\beta+(\alpha+1)\beta^2\big)
	\end{array}$$
The resulting code is $\mathcal C_{9}=\{N_{9}(R,S,C,E,a,b)\mid R,S,C,E,a,b\in F\}$, where
\begin{equation*}N_{9}(R,S,C,E,a,b)=\begin{pmatrix}R\!+\!{2}a\!+\!(2\alpha\!+\!2)b&{2\alpha}C\!+\!(\alpha\!+\!1)E\!+\!{2}a\!+\!{2\alpha}b&(2\alpha\!+\!2)C\!+\!(\alpha\!+\!2)E\!+\!b&(\alpha\!+\!1)S\\C&R\!+\!a&(\alpha\!+\!1)S\!+\!b&(2\alpha\!+\!2)E\\E&S&R&{2}C\\0&0&0&a\\0&0&0&b\end{pmatrix}.
\end{equation*}
A SageMath computation shows that the rank distribution of $\mathcal C_9$ is
\[
(A_0,A_1,A_2,A_3,A_4)=(1,0,0,72080,459360),
\]
\end{example}

\begin{remark}\label{rem:3-2} The proof strategy of Theorem \ref{thm:char3-general} is not applicable to \(q=3\), since \(\mathbb F_3=\{0,1,-1\}\), whereas the theorem requires some \(c\in F^*\setminus\{1,-1\}\). For this case, we find an optimal \([[3,3,3,5],6,3]_3\) Ferrers diagram rank-metric code  $\mathcal C_3$ by computer search, which is listed as follows:
	$$\mathcal C_3=\{M_3(a,b,c,d,e,f)\mid a,b,c,d,e,f\in \mathbb F_3\},$$ 
	 where
\[
M_3(a,b,c,d,e,f)=
\begin{pmatrix}
	b\!-\!e\!-\!f&c\!+\!e&\!-\!a\!+\!d&a\\
	a\!-\!e&d\!+\!f&b\!-\!c&\!-\!b\\
	d\!-\!e\!-\!f&b&\!-\!a\!-\!c&c\\
	0&0&0&e\\
	0&0&0&f
\end{pmatrix}.
 \] The rank distribution of $\mathcal C_3$ is
\[
(A_0,A_1,A_2,A_3,A_4)=(1,0,0,392,336),
\]
\end{remark}

\section{Proof of the main results when \texorpdfstring{$q$}{q} is even}\label{sec:even}
In this section, we prove that there exists an optimal \([[3,3,3,5],6,3]_q\) Ferrers diagram rank-metric code when $q$ is even.
While the construction of \([3\times 4,4,3]_q\) MRD codes in Theorem \ref{thm:odd-core} works for odd \(q\), it may fail to produce MRD codes in characteristic \(2\), as every element is a square in this case. Therefore, we need to find an alternative approach.

Throughout this section, \(q\) is assumed to be even.

\subsection{Construction of full row rank MRDs}
 Starting from the MRD codes obtained for odd \(q\) in Theorem \ref{thm:odd-core}, we try several modifications of the parameters in the construction and eventually obtain a new \([3\times 4,4,3]_q\) MRD code for even $q$.

\begin{theorem}\label{thm:char2-core} Let $t\in F$ and $n\in F $ such that \(\Tr_{F/\mathbb F_2}(n)=1\). 
For
$(R,S,C,E)\in F^4$, set 
	\begin{equation*}\label{eq:char2-core}
	\tilde M_{n,t}(R,S,C,E)=
		\begin{pmatrix}
			R&tC+nE&(t+n)C+ntE&nS\\
			C&R&nS&C+nE\\
			E&S&R+S&C
		\end{pmatrix}.
	\end{equation*}
Then	the matrix space \(\tilde {\cK}_{n,t}=\{\tilde M_{n,t}(R,S,C,E)\mid R,S,C,E\in F\}\) is a \([3\times4,4,3]_q\) MRD code.
\end{theorem}

\begin{proof} Obviously,  \(\dim \tilde{\mathcal{K}}_{n,t}=4\). 
	Let \(\mathbb E=F(\rho)\), where \(\rho^2+\rho+n=0\).  Since \(\Tr_{F/\mathbb F_2}(n)=1\), we get that $\rho \notin F$. Then $X^2+X+n$ is irreducible over $F$.
	Let  \((a,b,d)\in F^3\) such that $(a,b,d)\ \tilde M_{n,t}(R,S,C,E)=0$. Then we get the linear system 
	\begin{equation*}\label{eq:char2-H}
		H(a,b,d)\begin{pmatrix}
			R\\S\\C\\E
		\end{pmatrix}=\begin{pmatrix}
			a&0&b&d\\
			b&d&at&an\\
			d&nb+d&a(t+n)&ant\\
			0&an&b+d&bn
		\end{pmatrix}\begin{pmatrix}
		R\\S\\C\\E
	\end{pmatrix}=\begin{pmatrix}
	0\\0\\0\\0
\end{pmatrix}.
	\end{equation*}
Since $\rho+1$ is another root of $X^2+X+n$, we get that $\rho^q=\rho+1$. Set 	$U=N_{\mathbb E/F}(d+b\rho)$. Then we get 
	$$
	U=(d+b\rho)(d+b\rho)^q=(d+b\rho)(d+b(\rho+1))=d^2+bd+nb^2.
	$$
	Direct expansion using \(\rho^2+\rho+n=0\) gives
	\begin{equation*}\label{eq:char2-det}
		\det H(a,b,d) =U^2+a^2nU+a^4n^2(t^2+t+n)
		=N_{\mathbb E/F}\bigl(U+a^2n(t+\rho)\bigr).
	\end{equation*}

Now suppose that $\det H(a,b,d)=0$.	Then $U+a^2n(t+\rho)=0.$  Since $\rho\notin F$, we get  \(a=U=0\).  It follows that  \(d+b\rho=0\), which implies that $b=d=0$. Consequently, if 
$(a,b,d)\ne (0,0,0)$, then $\det H(a,b,d)\ne0$, which implies $R=S=C=E=0$. It follows that each nonzero codeword in $\tilde{\cK}_{n,t}$ has full row rank.
 Consequently,  $\tilde{\cK}_{n,t}$ is a \([3\times4,4,3]_q\) MRD code. \end{proof}

\subsection{The Orthogonality Conditions}

To perform the puncturing-inclusion construction, we write
\(\tilde M_{n,t}=[\tilde{A}_{n,t}\mid \tilde{w}]\), where
$$	\tilde{A}_{n,t}=\begin{pmatrix}
		R&tC+nE&(t+n) C+nt E\\
		C&R&n S\\ E&S&R+S
	\end{pmatrix},\qquad
	\tilde{w}=\begin{pmatrix}nS\\C+nE\\C\end{pmatrix}.
$$
Set  \(\tilde{\mathcal L}_{n,t}=\{\tilde{A}_{n,t}\mid R,S,C,E\in F\}\). 

Consider the  cubic extension \(K/F\) and let \((K,\circ)\) be a semifield.
Let \(\mathbf p=(p_1,p_2,p_3)\) and \(\mathbf q=(q_1,q_2,q_3)\) be two bases of $K$ over $F$.  Set 
\[
B_z=\Bigl(\Tr_{K/F}\big(z(p_i\circ q_j)\big)\Bigr)_{i,j},\qquad
\mathcal S=\{B_z:z\in K\}.
\]

 Set $\mathbf p=(1,x,y), q_1=1$.  Similar to the case that $q$ is odd, we get that \(\cS\subseteq \tilde{\mathcal L}_{n,t}^\perp\) if and only if 
\begin{align}
	1+x\circ q_2+y\circ q_3&=0,\label{eq:c2-R}\\
	nx\circ q_3+y\circ q_2+y\circ q_3&=0,\label{eq:c2-S}\\
	tq_2+(t+n)q_3+x&=0,\label{eq:c2-C}\\
	nq_2+ntq_3+y&=0.\label{eq:c2-E}
\end{align}
Let \(\delta=t^2+t+n\). Irreducibility of $X^2+X+n$ gives \(\delta\ne0\). Equations  \ref{eq:c2-C} and \ref {eq:c2-E} have the unique solution
\begin{equation}\label{eq:c2-r23}
	q_2=\frac{ntx+(t+n)y}{n\delta},\qquad
	q_3=\frac{nx+ty}{n\delta}.
\end{equation}
The coefficient matrix of $(x,y)\mapsto(q_2,q_3)$ has determinant $1/(n\delta)\ne0$. Consequently, whenever $(1,x,y)$ is an $F$-basis of $K$, so is $(1,q_2,q_3)$.
Substitute \ref{eq:c2-r23}
into \eqref{eq:c2-R} and \eqref{eq:c2-S}. Using the bilinearity of $\circ$, we obtain
\begin{equation}
	n\delta
	+nt(x\circ x)
	+(t+n)(x\circ y)
	+n(y\circ x)
	+t(y\circ y)
	=0,
	\label{4.22}
\end{equation}
and
\begin{equation}
	n(x\circ x)
	+t(x\circ y)
	+(t+1)(y\circ x)
	+y\circ y
	=0.
	\label{4.23}
\end{equation}
Multiplying \eqref{4.23} by $t$ and adding it to \eqref{4.22}, using $\operatorname{char}F=2$, we obtain
\[
n\delta+\delta(x\circ y)+\delta(y\circ x)=0.
\]
Since $\delta\neq0$, we get 
\begin{equation}
	x\circ y+y\circ x=n.
	\label{4.26}
\end{equation}
Substituting (\ref{4.26}) into \eqref{4.23}, we obtain
\begin{equation}y\circ y+n(x\circ x)+y\circ x=tn.
\end{equation}

Consequently, we have proved the following result.

\begin{proposition}\label{prop:char2-solver}
Let $(K,\circ)$ be a semifield, $t\in F$ and $n\in F $ such that \(\Tr_{F/\mathbb F_2}(n)=1\).  If there exists a basis $(1,x,y)$ of  $K$ over $F$ such that 
$$	x\circ y+y\circ x=n\quad \mbox{and} \quad 
	y\circ y+n(x\circ x)+y\circ x=tn,
	$$
	then 	there exists a semifield code $\mathcal S $  such that $
	\mathcal S\subseteq \tilde{\mathcal L}_{n,t}^{\perp}
	$.
\end{proposition}
\subsection{Proof of the main result for \texorpdfstring{$\operatorname{char}(F)=2$}{char(F)=2}}

Write $q=2^m$ with $m\ge3$, choose $n\in F$ with $\Tr_{F/\mathbb F_2}(n)=1$, and take \(c\in F^*\) satisfying \(c^3\ne1\). Such a \(c\) exists since $q-1>3$, so we can define a semifield $(K,\circ)$ as in Proposition \ref{bilinear2}.
 In characteristic \(2\), \(\Tr_{K/F}(1)=1\), so \(K=F\oplus W\), where $W=\mathrm{ker}\ \mathrm{Tr}_{K/F}$.  Let  $(u,v)$ be an $F$-basis of $W$. Then  $(1,u,v)$ is an $F$-basis of $K$.
It follows  from Proposition \ref{prop:trace-zero-master} and Lemma \ref{lem:trace-zero-master} that 
 $$\textstyle{u\circ v-v\circ u
=
\frac{3c(1-c^2)}{(c^2-c+1)^2}
\bigl(uv^q-u^qv\bigr)=\frac{c(1+c^2)}{(c^2+c+1)^2}
\bigl(uv^q-u^qv\bigr)\in F^*.}$$
Set $\eta=\frac{c(1+c^2)}{(c^2+c+1)^2}
\bigl(uv^q-u^qv\bigr),$ and $v_0=\frac{n}{\eta}v$. Then 
$u\circ v_0-v_0\circ u=n.$
Uniquely write
\begin{equation*}\label{eq:char2-Z}
	v_0\circ v_0+n(u\circ u)+v_0\circ u=z_0+z_1u+z_2v_0,\qquad z_i\in F.
\end{equation*}
Let
$
	x=u+z_2,\ y=v_0+z_1.
$
One can check that  $x\circ y+y\circ x=n$ and 
\[
y\circ y+n(x\circ x)+y\circ x
=z_0+z_1^2+nz_2^2+z_1z_2\in F.
\]
Setting 
$
	t=\frac{z_0+z_1^2+nz_2^2+z_1z_2}{n}
$
yields 	$y\circ y+n(x\circ x)+y\circ x=tn.$ 
It follows from Proposition \ref{prop:char2-solver}   that there exists a $[3\times 3,3,3]_q$ MRD code \(\cS\subseteq \tilde{\mathcal L}_{n,t}^\perp\). By Proposition \ref{prop:mrd-ferrers-lift}, there exists an optimal $[[3,3,3,5],6,3]_{q}$ Ferrers diagram rank-metric code. 

Consequently, we have proved the following theorem.

\begin{theorem}\label{thm:char2-general}
	For every $m\ge3$, there exists an optimal
	\[
	[[3,3,3,5],6,3]_{2^m}
	\]
	Ferrers diagram rank-metric code.
\end{theorem}

 We now illustrate Theorem \ref{thm:char2-general} with an example.
\begin{example}[An explicit code over $\F_{8}$]\label{ex:q8}
	
	Let $F=\F_2(\alpha)$ with $\alpha^3+\alpha+1=0$, and let $K=F(\beta)$ with $\beta^3+\beta+\alpha=0$. The displayed cubic has no root in $F$, so it is irreducible. Choose
	\[c=\alpha,\qquad u=\beta,\qquad v=\beta^2.\]
	Both $u$ and $v$ have trace zero, and $(1,u,v)$ is an $F$-basis of $K$. 
	Take $n=1$, for which $\Tr_{F/\F_2}(n)=1$. 
	% Set $t=\alpha^6.$ Then $\delta=t^2+t+n=\alpha^3.$
	Following the proof of Theorem \ref{thm:char2-general}, we obtain two desired bases:
	$$\mathbf p=(1,\alpha^{4}+\beta,\alpha^{4}+{\alpha^{2}}\beta^2),\quad 
	\mathbf q=(1,\alpha+{\alpha^{3}}\beta+{\alpha}\beta^2,\alpha^{3}+{\alpha^{4}}\beta+{\alpha^{5}}\beta^2)$$
	The resulting code is $\mathcal C_{8}=\{N_{8}(R,S,C,E,a,b)\mid R,S,C,E,a,b\in F\}$, where\begin{equation*}N_{8}(R,S,C,E,a,b)=\begin{pmatrix}R\!+\!{\alpha^{5}}a\!+\!b&{\alpha^{6}}C\!+\!E\!+\!b&{\alpha^{2}}C\!+\!{\alpha^{6}}E\!+\!{\alpha^{5}}a&S\\C&R&S\!+\!a&C\!+\!E\\E&S\!+\!b&R\!+\!S&C\\0&0&0&a\\0&0&0&b\end{pmatrix}.\end{equation*}
	A SageMath computation shows that the rank distribution of $\mathcal C_8$ is
	\[
	(A_0,A_1,A_2,A_3,A_4)=(1,0,0,40887,221256),
	\]

\end{example}
\begin{remark}\label{rem:4-1}The proof of Theorem \ref{thm:char2-general} requires the existence of some \(c\in F^*\) such that \(c^3\neq 1\). This condition fails for \(q=2,4\), so the proof strategy used there is not applicable to these two cases. For \(q=2,4\), we find optimal \([[3,3,3,5],6,3]_q\) Ferrers diagram rank-metric codes by computer search, which are listed as follows.
  
  \begin{itemize}
  \item[(1)] The code $$\mathcal C_2=\{M_2(a,b,c,d,e,f)\mid a,b,c,d,e,f\in \mathbb F_2\}$$ is an  optimal 	$[[3,3,3,5],6,3]_{2}$  Ferrers diagram rank-metric code, where
  $$M_2(a,b,c,d,e,f)=
  \begin{pmatrix}
  	a\!+\!b\!+\!d&d\!+\!e&b\!+\!e\!+\!f&a\!+\!b\!+\!c\!+\!d\\
  	a\!+\!b\!+\!c\!+\!f&a\!+\!b&a\!+\!e&a\!+\!d\\
  	b&c&d&a\\
  	0&0&0&e\\
  	0&0&0&f
  \end{pmatrix}.
  $$
  The rank distribution of $\mathcal C_2$ is
  \[
  (A_0,A_1,A_2,A_3,A_4)=(1,0,0,57,6),
  \]
Note that an optimal \([[3,3,3,5],6,3]_2\) Ferrers diagram rank-metric code was first given in \cite{Etzion2016}, which  is different from $\mathcal C_2$  given here.
\item[(2)]
Let
$
F=\F_2[\alpha]/(\alpha^2+\alpha+1).
$ Then the code $$\mathcal C_4=\{M_4(a,b,c,d,e,f)\mid a,b,c,d,e,f\in \mathbb F_4\}$$ is an  optimal 	$[[3,3,3,5],6,3]_{4}$  Ferrers diagram rank-metric code, where
\[
M_4(a,b,c,d,e,f)=
\begin{pmatrix}
	\alpha^2b\!+\!f&c\!+\!\alpha^2d&a\!+\!d\!+\!\alpha^2e&\alpha^2a\\
	a&b&c&a\!+\!b\!+\!\alpha^2d\\
	d&\alpha^2a\!+\!\alpha d\!+\!e&\alpha b\!+\!f&\alpha^2b\!+\!\alpha c\!+\!d\\
	0&0&0&e\\
	0&0&0&f
\end{pmatrix}.
\]   The rank distribution of $\mathcal C_4$ is
\[
(A_0,A_1,A_2,A_3,A_4)=(1,0,0,1515,2580),
\]
\end{itemize}
\end{remark}

\section{Applications of the main results and a binary construction}\label{sec:general}

\subsection{Applications of the main results}

Combining Theorems \ref{thm:charno3-general},\ref{thm:char3-general}, \ref{thm:char2-general} and Remarks \ref{rem:3-2},\ref{rem:4-1}, we get the following main theorem:

\begin{theorem}\label{thm:all}
	For every prime power $q$, there exists an optimal
	$[[3,3,3,5],6,3]_q$ Ferrers diagram rank-metric code.
\end{theorem}

\begin{remark}
The optimal code in Theorem~\ref{thm:all} lies outside the scope of the following standard constructions.
\begin{itemize}
\item[(1)] \emph{Extension-field-linear MRD subcodes.}
For $\mathcal F_0=[3,3,3,5]$ and distance $3$, the ambient size is $5\times4$ and $\ell=4-3+1=2$. We have
\[
c_\ell=c_{\ell+1}=3<5,\qquad \gcd(3,5)=1,\qquad
\nu_{\min}(\mathcal F_0,3)=\nu_0(\mathcal F_0,3)=3+3=6.
\]
By the coprimality obstruction in~\cite[Corollary~4.2]{AntrobusGluesingLuerssen2019}, an optimal code on $\mathcal F_0$ cannot be realized as an $\mathbb F_q$-subspace of an $\mathbb F_{q^5}$-linear $[5\times4,10,3]_q$ MRD code under coordinate expansion. This excludes the systematic construction based on such an extension-field-linear MRD code.

\item[(2)] \emph{Monotone diagrams and diagonal MDS constructions.}
In the square convention of~\cite{NeriStanojkovski2024}, pad the diagram with an empty column to obtain $D=[0,3,3,3,5]$. Its repeated positive heights below $5$ violate the monotonicity condition in~\cite[Definition~2.14]{NeriStanojkovski2024}. Its adjoint $D^{\text{T}}=[1,1,4,4,4]$ is likewise not monotone. Moreover, the diagonal lengths in Definition~2.12 of that reference are $(3,4,3,2,1)$, so the diagonal MDS construction yields dimension
\[
\sum_{i=1}^5\max\{0,|D\cap\Delta_i^5|-3+1\}=1+2+1=4<6.
\]
Thus $(D,3)$ is not MDS-constructible, and these existence results do not yield an optimal code here.

\item[(3)] \emph{The standard equal-dimension combination construction.}
The obstruction in~\cite[Lemma~2.13]{LiuChangFeng2019} applies to the combination construction in Theorem~2.12 of that reference, which restates~\cite[Theorem~9]{Etzion2016}. With the zero-based heights $(\gamma_0,\gamma_1,\gamma_2,\gamma_3)=(3,3,3,5)$, we have
\[
\max_{2\le i\le3}(\gamma_i-\gamma_{i-1})=2<\nu_0=6=\nu_{\min}(\mathcal F_0,3).
\]
Hence this particular combination construction cannot produce an optimal code of dimension $6$ on $\mathcal F_0$. 
\end{itemize}
\end{remark}

As an application of Theorem \ref{thm:all},  we give a construction of  optimal Ferrers diagram rank-metric codes for an infinite family of diagrams by the combination construction given in \cite{Etzion2016}. 

\begin{theorem}\label{thm:general}
	Let $n\ge1$ and write $n=6t+r$, where $t\ge0$ and $0\le r\le5$. Define
	\begin{align*}
		\mathcal F_n=\big[&3,3,3,5,
		\underbrace{11,\cdots,11}_{6},
		\underbrace{17,\cdots,17}_{6},\cdots,
		\underbrace{5+6t,\cdots,5+6t}_{6},\underbrace{5+6(t+1),\cdots,5+6(t+1)}_{r}\big],
	\end{align*}

\medskip
\noindent 	where the blocks indexed by $1,\cdots,t$ are omitted when $t=0$, and the last block is omitted when $r=0$.
Then	for every prime power $q$, there exists an optimal
	$[\mathcal F_n,6,n+3]_q$ Ferrers diagram rank-metric code.
\end{theorem}
\begin{proof}
A direct count gives $\nu_0=\nu_1=\nu_2=6$ and $\nu_i>6$ for $3\le i\le n+2$, so $\nu_{\min}(\mathcal F_n,n+3)=6$.
Let $\mathcal C_0$ be the six-dimensional seed code in Theorem~\ref{thm:all}, and let $\mathcal U$ be a $[6\times6,6,6]_q$ MRD code. If $1\le r\le5$, let $\mathcal V_r$ be a $[6\times r,6,r]_q$ MRD code. Such codes are obtained from multiplication in $\mathbb F_{q^6}$, using all six coordinate columns for $\mathcal U$ and the first $r$ columns for $\mathcal V_r$.
Choose $\mathbb F_q$-linear isomorphisms
\[
\Phi_0:\mathbb F_q^6\longrightarrow\mathcal C_0,\qquad
\Phi_{\mathrm{full}}:\mathbb F_q^6\longrightarrow\mathcal U,
\]
and, when $r>0$, $\Phi_{\mathrm{tail},r}:\mathbb F_q^6\longrightarrow\mathcal V_r$.
For $v\in\mathbb F_q^6$, set
\[
B_n(v)=\operatorname{diag}\bigl(\Phi_0(v),
\underbrace{\Phi_{\mathrm{full}}(v),\ldots,\Phi_{\mathrm{full}}(v)}_{t},
\Phi_{\mathrm{tail},r}(v)\bigr),
\]
%where the final block, including its six rows, is omitted when $r=0$.
Obviously, the support of $B_n(v)$ is contained in $\mathcal F_n$. Since its seed block is $\Phi_0(v)$, the encoding map is injective; hence the code $\mathcal C=\{B_n(v):v\in\mathbb F_q^6\}$ has dimension $6$.
Every nonzero codeword satisfies
\[
\rank B_n(v)=\rank\Phi_0(v)+6t+r\ge3+6t+r=n+3.
\]
Since $\Phi_0$ is onto and the seed code has minimum rank distance $3$, equality occurs. Therefore $\mathcal C$ is an optimal $[\mathcal F_n,6,n+3]_q$ Ferrers diagram rank-metric code.
\end{proof}

\subsection{A binary construction of  optimal Ferrers diagram rank-metric codes}

In this subsection, we construct optimal Ferrers diagram rank-metric codes for an infinite family of diagrams over \(\mathbb{F}_2\) directly by the puncturing-inclusion construction. This construction is independent of the proof of Theorem \ref{thm:all}.

\begin{theorem}\label{thm:binary}
		Let $m\ge3$, $1\le t\le m-1$ and $$
	\mathcal F_{m,t}=[\underbrace{m,m,\ldots,m}_{m\text{ columns}},m+t].$$

	\medskip
	\noindent Then there  exists an optimal $[\mathcal F_{m,t},m+t+1,m]_2$  Ferrers diagram rank-metric code.
\end{theorem}

\begin{proof}
Fix an ordered $\mathbb{F}_{2}$-basis
$\mathcal B=(e_0,e_1,\cdots,e_{m-1})$ of $\mathbb{F}_{2^m}$ with $e_0=1$,
write $[x]_{\mathcal B}$ for the coordinate column of $x\in \mathbb{F}_{2^m}$.
Choose $\lambda\in \mathbb{F}_{2^m}$ such that $\Tr_{\mathbb{F}_{2^m}/\mathbb{F}_{2}}(\lambda)=1$.
Let $V_t=\operatorname{span}_{\mathbb{F}_{2}}\{e_0,e_1,\cdots,e_t\}$. For $a\in V_t$ and $b\in \mathbb{F}_{2^m}$, write uniquely
\[
a=a_0e_0+\cdots+a_te_t,\qquad a_i\in \mathbb{F}_{2},
\]
and define the $\mathbb{F}_{2}$-linear maps
\[
T_{a,b}(x)=ax+bx^2,\qquad r_b(x)=\Tr_{\mathbb{F}_{2^m}/\mathbb{F}_{2}}(\lambda bx).
\]
Let $M(a,b)$ be the matrix of $T_{a,b}$ in $\mathcal B$, and set
\[
\rho_b=\bigl(\Tr_{\mathbb{F}_{2^m}/\mathbb{F}_{2}}(\lambda be_0),\cdots,
\Tr_{\mathbb{F}_{2^m}/\mathbb{F}_{2}}(\lambda be_{m-1})\bigr),
\qquad h(a)=(a_1,\cdots,a_t)^T.
\]
Thus $M(a,b)[x]_{\mathcal B}=[T_{a,b}(x)]_{\mathcal B}$ and $\rho_b[x]_{\mathcal B}=r_b(x)$.
Define
\begin{equation*}\label{eq:binary-tail-encoding}
	C_t(a,b)=\begin{pmatrix}
		M(a,b)^T&\rho_b^T\\
		0_{t\times m}&h(a)
	\end{pmatrix},
	\qquad
	\mathcal C_{m,t}=\{C_t(a,b):a\in V_t,\ b\in \mathbb{F}_{2^m}\}.
\end{equation*}
We claim that $	\mathcal C_{m,t}$ is an  $[\mathcal F_{m,t},m+t+1,m]_2$  Ferrers diagram rank-metric code.
By construction, $	\mathcal C_{m,t}$ is supported on $\mathcal F_{m,t}$ and the map $f: (a,b)\mapsto C_t(a,b)$ is $\mathbb F_2$ -linear. 

{\it The dimension:} Let $a\in V_t,\ b\in \mathbb{F}_{2^m}$.
		If $C_t(a,b)=0$, then $T_{a,b}$ is identically zero. Substituting $x=1$ gives $a=b$. Since $m\ge3$, there exists $x\in \mathbb{F}_{2^m}\setminus \mathbb{F}_{2}$ such that 
	\[
	0=ax+bx^2=b(x+x^2).
	\]
	As $x\notin\{0,1\}$, we have $x+x^2\ne0$, so $b=a=0$. Thus $f$ is  injective and
therefore	\[
	\dim_{\mathbb{F}_{2}}\mathcal C_{m,t}=\dim_{\mathbb{F}_{2}} V_t+\dim_{\mathbb{F}_{2}} \mathbb{F}_{2^m}=m+t+1.
	\]

{\it The distance:}	For $(a,b)\ne(0,0)$, if $b=0$ then $T_{a,0}$ is invertible. If $b\ne0$, the factorization $T_{a,b}(x)=x(a+bx)$ shows that its kernel is contained in $\{0,a/b\}$, and hence has $\mathbb{F}_{2}$-dimension at most one. Consequently
	\begin{equation}\label{eq:binary-core-rank}
		\rank M(a,b)\ge m-1.
	\end{equation}
	If $h(a)\ne0$, the bottom part of the last column is nonzero, whereas the bottom parts of the first $m$ columns are zero. The last column is therefore independent of the first $m$ columns, and~\eqref{eq:binary-core-rank} gives
	\[
	\rank C_t(a,b)=\rank M(a,b)+1\ge m.
	\]
	If $h(a)=0$, then $a=a_0\in \mathbb{F}_{2}$. It suffices to show that $x\mapsto(T_{a,b}(x),r_b(x))$ is injective, since the top rectangular block is the transpose of its matrix:
	\[
	[M(a,b)^T\mid\rho_b^T]
	=\begin{pmatrix}M(a,b)\\\rho_b\end{pmatrix}^{T}.
	\]
	When $a=0$, $b\ne0$, the map $T_{0,b}(x)=bx^2$ is invertible; when $a=1$, $b=0$, the map $T_{1,0}(x)=x$ is invertible. In the remaining case $a=1$, $b\ne0$, we have
	\[
	\ker T_{1,b}=\{0,b^{-1}\},\qquad
	r_b(b^{-1})=\Tr_{\mathbb{F}_{2^m}/\mathbb{F}_{2}}(\lambda b b^{-1})=\Tr_{\mathbb{F}_{2^m}/\mathbb{F}_{2}}(\lambda)=1.
	\]
	Thus $\ker T_{1,b}\cap\ker r_b=\{0\}$. It follows that every nonzero codeword with zero bottom block  has rank  $m$.
	Consequently,  $d_R(\mathcal C_{m,t})=m$.
	
	\emph{Attainment of the dimension bound.}
	For the Ferrers diagram $\mathcal F_{m,t}$
 with rank-metric distance $m$, one can check that $\nu_{m-1}=m+t+1$, $\nu_0=2m$ and $\nu_{j}\geq m+t+1$
 for $1\leq j\leq m-2$. Since $m>t$,  we have  $\nu_{min}(\mathcal F_{m,t},m)=m+t+1$.
  Consequently, $\mathcal C_{m,t}$ is an optimal $[\mathcal F_{m,t},m+t+1,m]_2$ Ferrers diagram rank-metric code.
\end{proof}

\section{Conclusions and discussions}\label{sec:conclusion}

In this paper, we have constructed an optimal \([[3,3,3,5],6,3]_q\) Ferrers diagram rank-metric code for every prime power \(q\), which answers an open problem posed by Etzion et al.~\cite{Etzion2016}. 
We have also extended this result in two directions. First, we have constructed optimal Ferrers diagram rank-metric codes for an infinite family of Ferrers diagrams, obtained by a block combination construction based on the \([[3,3,3,5],6,3]_q\) seed. Second, independently of the preceding approach, we have directly constructed optimal binary Ferrers diagram rank-metric codes for another infinite family of diagrams, using the puncturing-inclusion construction. The first family extends the \([3,3,3,5]\) seed, while the second gives an independent binary construction.

It was proved in \cite{Prajapati2026} that the Ferrers diagram \([4,4,4,4,6]\) does not admit an optimal linear rank-metric code over \(\mathbb F_2\) at rank distance \(3\). This prevents an unconditional extension to all finite fields for the entire family \([\underbrace{n-1,\ldots,n-1}_{n-1},n+1]\); existence for other fields and parameters must be studied separately. Nevertheless,  our method may be a promising approach to determining which diagrams of this type admit optimal Ferrers diagram rank-metric codes, and due to their crucial role in the multilevel construction of constant dimension codes, this remains an interesting and worthwhile problem.

We end with two open problems.

\begin{problem}
	Can the construction of the optimal binary Ferrers diagram rank-metric codes, obtained directly by the puncturing-inclusion construction, be extended to general finite fields \(\mathbb{F}_q\)? More generally, for which prime powers \(q\) does an optimal code for the corresponding infinite family of diagrams exist?
\end{problem}

\begin{problem}
	How can the construction for the diagram \([3,3,3,5]\) be extended to other irreducible Ferrers diagrams? In particular, which irreducible diagrams admit optimal Ferrers diagram rank-metric codes, and can the Delsarte dual and puncturing-inclusion techniques developed here be adapted to them?
\end{problem}

\medskip
\noindent{\bf Acknowledgements} The authors  would like to thank 
 Wei  Hu  for  useful discussions. Dengming Xu  is supported by the Natural Science Foundation of Tianjin (Grant No.23JCQNJC00050).

\end{document}